\documentclass[aoas,preprint]{imsart}

\RequirePackage[OT1]{fontenc}
\RequirePackage{amsthm,amsmath}
\usepackage{mathtools}
\usepackage[english]{babel}

\RequirePackage[square]{natbib}
\usepackage{hyperref}

\usepackage{verbatim}
\usepackage{graphicx}
\usepackage{float}
\usepackage[linesnumbered,lined,boxed,commentsnumbered]{algorithm2e}

\usepackage{cleveref}
\Crefname{algocf}{Algorithm}{Algorithms}

\graphicspath{ {images/} }

\usepackage{pgfplots}
\usepgfplotslibrary{groupplots}

\pgfplotsset{compat=1.5}

\newtheorem{proposition}{Proposition}

\startlocaldefs
\numberwithin{equation}{section}
\theoremstyle{plain}

\endlocaldefs

\begin{document}

    \begin{frontmatter}
        \title{Efficient Construction of Test Inversion Confidence Intervals Using Quantile Regression, with Application to Population Genetics}
        \runtitle{Efficient Construction of Confidence Intervals}

        \begin{aug}
            \author{\fnms{Eyal} \snm{Fisher}\ead[label=e1]{eyalfisher@mail.tau.ac.il}},

            \author{\fnms{Regev} \snm{Schweiger}\ead[label=e2]{regevs@gmail.com}}
            \and
            \author{\fnms{Saharon} \snm{Rosset}
                \ead[label=e3]{saharon@post.tau.ac.il}}

            \affiliation{Tel Aviv University}

%
        \end{aug}

        \begin{abstract}
            ~Modern problems in statistics tend to include estimators of high computational complexity and with complicated distributions. Statistical inference on such estimators usually relies on asymptotic normality assumptions, however, such assumptions are often not applicable for available sample sizes, due to dependencies in the data and other causes. A common alternative is the use of re-sampling procedures, such as the bootstrap, but these may be computationally intensive to an extent that renders them impractical for modern problems. In this paper we develop a method for fast construction of test-inversion bootstrap confidence intervals. Our approach uses quantile regression to model the quantile of an estimator conditional on the true value of the parameter, and we apply it on the Watterson estimator of mutation rate in a standard coalescent model. We demonstrate an improved efficiency of up to 40\% from using quantile regression compared to state of the art methods based on stochastic approximation, as measured by the number of simulations required to achieve comparable accuracy.
        \end{abstract}

    \end{frontmatter}

    \section{Introduction}

    When applying statistical modeling methodologies to modern problems in scientific domains, characteristics often encountered include high computational complexity of the estimation process, and statistical complexity of the resulting estimators. It is therefore important to develop statistical inference approaches that can be applied to cases where distributions of estimators are hard to derive analytically, and resampling-based approaches are computationally prohibitive, and have to be efficiently implemented.

    This scenario is common in the area of population genetics, in which probabilistic models are constructed and fitted in order to examine certain characteristics of a population, such as the population mutation rate and recombination rate, which are of high scientific interest~\citep{Myers2013,Scally2012,Sigurgardottir2000} and are commonly estimated using coalescence models~\citep{Hey1997,Narasimhan2016,Palamara2015}.

    The resulting estimators tend to have a complicated distribution, and simulations can involve simulating the evolution of millions of DNA bases, in thousands of individuals for many generations, and therefore are computationally costly.

    We concentrate on the challenge of deriving confidence intervals for the parameter estimates in such problems. When asymptotic parametric approaches do not exist or are unreliable for given sample sizes, it is common to use resampling-based approaches, usually based on the bootstrap and its variants, and methods based on non-parametric bootstrap have been widely used and studied~\citep{Felsenstein1985,Efron2016}. However, they often encounter difficulties when the statistical setting is too complex to allow the sampling schemes required for implementing non-parametric bootstrap~\citep{Efron1996,Efron2003}.

    An important approach to constructing confidence intervals is the method known as ``test inversion'', which uses the duality between hypothesis testing and confidence intervals, defining the $1-\alpha$ confidence interval as the set of parameter values for which the test of $H_0: \theta=\theta_0$ is not rejected at level $\alpha$. When the relevant rejection regions cannot be analytically defined, test inversion is implemented through simulation, essentially using a parametric-bootstrap approach. The search for the confidence interval endpoint is often estimated through the efficient Robbins-Monro (RM) procedure~\citep{Carpenter1999,PaulH.GarthwaiteandStephenT.Buckland2016}, however, the RM algorithm is inaccurate for extreme quantile and does not achieve optimal convergence rates~\citep{Wetherill2016,Young2016}.

    In this paper we describe three methodologies for finding the correct endpoint efficiently. The commonly used approach of \cite{PaulH.GarthwaiteandStephenT.Buckland2016} (RM), a new adaption of the method that is based on an RM algorithm for binary data developed by \cite{Joseph2004}, to which we refer as Binary Robbins-Monro (BRM), and a novel method, that uses an adaptive quantile regression (AQR) and inverts the estimated quantiles to determine the endpoint. We compare the approaches analytically and empirically and show that the latter is more efficient when only a small number of simulations can be sampled from the model due to computational limitations.

    Finally, we apply and compare the suggested methodologies for the purpose of constructing confidence intervals for the mutation rate parameter under a standard coalescent model, using the Watterson estimator for the population mutation rate, for which the convergence to the normal approximation is known to be slow and therefore inaccurate for moderate sample size~\citep{Klein1999}. Estimates of the human mutation rate are frequently used in order to date events in our population history~\citep{Scally2012,Schiffels2014}. Using a different estimate of the mutation rate can substantially influence such analysis~\citep{Sigurgardottir2000}, thus, it is important to accurately measure the uncertainty of such estimates.

    \section{Why Test Inversion?}

    This section is devoted for a brief discussion about the different procedures for the construction of bootstrap confidence intervals and their drawbacks. For a more complete discussion see \cite{Carpenter2000}.

    The most common family of bootstrap methods is that of the pivotal methods, which includes the basic bootstrap and the bootstrap-t.
    These methods are very similar to the classical methods for construction of confidence intervals, but when using them we replace $w_\alpha$, the $\alpha$ percentile of the unknown reference distribution with $w^{\ast}_\alpha$, the $\alpha$ percentile of the bootstrap distribution. We can then use studentization to reduce the coverage error.

    Un-studentized pivotal intervals tend to be inaccurate. The studentized confidence intervals may contain invalid values, and rely on knowledge of the variance of the estimator or a computationally heavy second-level bootstrap for the estimation of the variance.

    The second family is the percentile family. Here we aim to take the $1-\alpha$ empirical percentile of the bootstrap distribution to be the upper end of our interval. The method is simple, cannot contain invalid values and is transformation respecting.

    However, its justification depends on the existence of a function $g(\cdot)$ such that
    $g(\hat{\theta}^{\ast}) - g(\hat{\theta}) \sim g(\hat{\theta}) - g(\theta) \sim N(0, \sigma^2)$ .
    In many applications such function does not exist and we get a substantial coverage error. Improvements to the methods, such as the BCa method, solve part of the problem but add a lot of complexity, and still rely on the existence of a function that cannot easily be shown to exist.

    The third family, the test inversion family, offers a simple alternative to the pivotal and percentile methods which can be implemented whenever parametric bootstrap sampling is possible. It does not suffer from the issues mentioned above, does not rely on any implicit assumptions, and in the absence of nuisance parameters has no bootstrap error, which can be significant for the other methods if proper re-sampling is complex due to dependencies in the data. It is, however, computationally demanding, which make it especially important to develop methods that allow for faster calculation of bootstrap test inversion confidence intervals.

    \section{Methodology}

    We start by describing the test inversion methodology. Let $X$ be a random variable with some density $f_\theta$ that depends on an unknown parameter of interest $\theta$, and let $\hat{\Theta}(X)$ be an estimator of $\theta$.
    If $\hat{\Theta}(X)$ is stochastically increasing with $\theta$ and $\hat{\theta} = \hat{\Theta}(x)$ is the estimate of $\theta$ based on a sample $x$, then the correct endpoint $U$ of a one sided (1-$\alpha$) confidence interval for $\theta$, satisfies:
    $$P_{\theta=U}(\hat{\theta} < \hat{\Theta}(X)) = 1-\alpha .$$
    In this setting, $U$ is the smallest value for which we would have rejected the hypothesis $H_0: \theta=U$ in favor of $H_1: \theta<U$ in an $\alpha$ level test.

    Our problem is now focused on finding the point $U$ with this property.
    Since the distribution of the estimator is unknown, $U$ cannot be inferred analytically and has to be estimated by a Monte-Carlo simulation.
    The straight forward way to carry such a simulation is to interpolate sample quantile of $\hat{\Theta}(X)$, as in \cite{Schweiger2015}. This approach is briefly described in \Cref{algo_naive}, resulting in $\hat{U}$ - an estimate of $U$.
    However, this approach is computational inefficient in a way that makes it impractical for many problems, as it entails simulating many samples from the distribution in areas that are not the interest of the analysis.

    \IncMargin{1em}
    \begin{algorithm}
        Choose a grid of points $\underline{\theta} = (\theta_1, \theta_2, ..., \theta_n)$

        \For{$i\leftarrow 1$ \KwTo $n$}{
            Set $\theta \leftarrow \theta_i$

            \For{$j\leftarrow 1$ \KwTo $B$}{
                Sample $x_i^j$ from $f_{\theta}$

                $\hat{\theta}_i^j \leftarrow \hat{\Theta}(x_i^j)$
            }

            Calculate $q_i^\alpha$, the $\alpha$ sample quantile of $\hat{\theta}_i^1, ..., \hat{\theta}_i^B$
        }

        Find $k$ such that $q_k^\alpha < \hat{\theta} < q_{k+1}^\alpha$

        $\hat{U} \leftarrow \frac{(\hat{\theta} - q_k^\alpha)\cdot(\theta_{k+1} - \theta_k)}{q_{k+1}^\alpha - q_k^\alpha} + \theta_k$

        \caption{Finding the upper endpoint of a $100(1-\alpha)\%$ confidence interval by interpolating sample quantiles}
        \label{algo_naive}
    \end{algorithm}

    It was shown by \cite{Carpenter2000} that a more efficient approach is to use stochastic approximation to solve $M(\theta) = 1-\alpha$ for $\theta$, where $M(\theta) = P_\theta(\hat{\theta} < \hat{\Theta}(X))$. The value of $M(\theta)$ is unknown, but simulation can be carried for a given $\theta$ to get a noisy observation from the function.

    \subsection{Robbins-Monro algorithm}

    Given a function $M(\theta)$ such as the one described above, it is shown by \cite{Robbins1951} that the sequence
    $$
    x_{n+1} = x_n - a_n(y_n-b_n) ,
    $$
    where $y_n$ is the nth noisy observation, converges to the solution of $M(\theta) = 1-\alpha$,
    for every sequence $a_n$ for which:
    $$
    b_n = \alpha, \, \sum_{n=1}^{\infty}a_n = \infty, \,\sum_{n=1}^{\infty}a_n^2 < \infty .
    $$

    In our case
    $$
    y_n =
    \begin{dcases}
    1 ,& \text{if } \hat{\theta}_i > \hat{\theta},\\
    0,              & \text{otherwise}.
    \end{dcases}
    $$
    \cite{PaulH.GarthwaiteandStephenT.Buckland2016} described the use of the Robbins-Monro process for finding the endpoint of the confidence interval.
    We let $U_i$ be the current estimate of the endpoint and $\hat{\theta}_i$ be the current estimate of $\theta$, based on a random sample taken with $\theta=U_i$. In each step we update $U_i$ in the following manner:
    $$
    U_{i+1} =
    \begin{dcases}
    U_i - \frac{c\alpha}{i},& \text{if } \hat{\theta}_i > \hat{\theta},\\
    U_i + \frac{c(1-\alpha)}{i},              & \text{otherwise}.
    \end{dcases}
    $$
    Where $c$ is a step size constant. The procedure is shown to be fully asymptotically efficient (the variance of $U_i$ meets Cramer-Rao lower bound for a-parametric estimators) if $c$ is set to be 1/$M^{'}(U)$, the inverse of the slope of $M$ at the endpoint. However, neither $M$ or $U$ are known, so $c$ is estimated adaptively, using $U_i$ in place of $U$ and setting it to twice the optimal value for the normal distribution, as using a bigger than optimal step size is less damaging to convergence rate than a too small constant.

    The next section describes an adaptation of this process that makes a different choice of $a_n$ and $b_n$ in order to obtain better results for extreme quantiles.

    \subsubsection{Binary Robbins-Monro}

    As mentioned above, if the optimal step size constant is known, the Robbins-Monro procedure is fully asymptotically efficient. However, it was empirically shown to work poorly for extreme quantiles. In order to improve the convergence of the process, \cite{Joseph2004} suggested a modified algorithm that takes advantage of the fact that in order to search for the quantile we use the binary response:
    $$
    Y_{i} =
    \begin{dcases}
    1,& \text{if } \hat{\theta}_i > \hat{\theta},\\
    0,              & \text{otherwise}.
    \end{dcases} ,
    $$
    Next, he defines $M(x)$ to be the probability that $Y_{i} = 1$ conditioned on $\theta=x$. The goal is again to find $U$ for which $M(U) = 1-\alpha$.

    For convenience, we denote $M(x) = M(x- U)$ so that $M(0) = 1-\alpha$.
    \cite{Joseph2004} continues in a Bayesian framework, assuming a prior distribution for $U$, for which $\mathrm{E}(\Theta) = x_1, \mathrm{Var}(\Theta) = \sigma_1^2$. Denoting $Z_n = x_n - \Theta$, the goal now is to find the sequences $a_n, b_n$ for which $Z_n \rightarrow 0$ at the fastest rate.
    Using a normal approximation for $M(x)$, he showed that the optimal procedure sets
    $$
    x_{n+1} = x_n - \frac{c_n}{\beta b_n (1-b_n)}(y-b_n) ,
    $$
    where
    $$
    c_n = \frac{v_n}{(1+v_n)^\frac{1}{2}}\phi\left(\frac{\Phi^{-1}(\alpha)}{(1+v_n)^\frac{1}{2}}\right),
    b_n = \Phi \left(\frac{\Phi^{-1}(\alpha)}{(1+v_n)^\frac{1}{2}} \right),
    v_{n+1} = v_n - \frac{c_n^2}{b_n(1-b_n)}
    $$
    with $v_1 = \beta^2\tau_1^2$, $\beta = \frac{M'(0)}{\phi(\Phi^{-1}(\alpha))}$. \\

    Interestingly, this procedure results in a sequence, $b_n$, that starts between $\alpha$ and 0.5, and converges to $\alpha$ as $n$ increases, essentially starting the search by searching a less extreme quantile and advancing towards $\alpha$ as more data is gathered.

    Under the optimal assignment of $M'(0)$, this modified procedure is more accurate than regular Robbins-Monro for extreme quantiles and a small sample size. However, neither M nor the correct endpoint are known when aiming to construct confidence intervals for complicated distributions, therefore M has to be replaced by an approximation and $M'(0)$ has to be estimated adaptively from the data. In order to do so, we follow the suggestion by \cite{PaulH.GarthwaiteandStephenT.Buckland2016}, setting it to be proportional to the distance from $\hat{\theta}$:
    $$M'(0) = \frac{1}{k(U_i - \hat{\theta}(y))}, k = \frac{1}{z_\alpha(2\pi)^{-1/2}e^{-z_\alpha^2/2}} .$$

    In addition, the modified algorithm requires the specification of the prior mean and variance.
    In all the simulations that are carried in this paper we choose the prior mean by randomly picking a point from a $N(U, 1)$ distribution, we set the prior variance to be 1, and estimate $M'(0)$ adaptively. We do this identically for the binary and regular Robbins-Monro.

    Both the binary and regular Robbins-Monro procedures still entail some loss of information, as we don't make use of samples that were taken along the process for the determination of $U$ in the later steps and we don't make use of the actual values of the estimates. Finally, the process forces us to search separately for the upper and lower limits, and re-sample from scratch for each of them.  We offer to fix this loss of information by considering the stochastic function $\hat{\Theta}(X)$, which returns an estimate of $\theta$ for a re-sample (from the model with a given $\theta$). We assume a parametric model on the quantiles of this function and use quantile regression to find where the $\alpha$ quantile meets $\hat{\theta}$.

    \subsection{Quantile regression}

    Quantile regression~\citep{Koenker2013} is a type of regression analysis in which we aim to model the conditional quantile of a distribution, instead of the conditional mean.
    This is done by assuming the model $Q_\alpha(Y|X=x) = \beta_0 + \beta_1x$ and minimizing a loss function with respect to $\beta_0, \beta_1$. Unlike the normal linear regression, the loss function for quantile regression is asymmetric and the solution is achieved numerically, as there is no closed-form solution.  Define:
    $$
    \rho_\tau(u) = u(\tau - I(u < 0)) ,
    $$
    The quantile loss is given by:
    $$
    L(\beta_0, \beta_1) = \sum_{i=1}^{n}\rho_\tau(y_i - (\beta_0 + \beta_1x_i)) .
    $$

    The methodology that we suggest revolves around modeling the quantile of the distribution of the bootstrap estimators conditional on the true value of $\theta$ as linear, and solving the resulting equation $\beta_0 + \beta_1
    \hat{U} = \hat{\theta}$ for $\hat{U}$.

    \subsubsection{Asymptotics of the quantile regression procedure}

    For the simplest case, we assume that the linear model of the conditional quantile is true and that the density of the quantiles is fixed. That is, if we denote the density of $\hat{\Theta}(X)$ by $g_\theta(x)$, and we denote the $\tau$ quantile of $g_\theta(x)$ by $\xi_\theta(\tau)$, then $g_\theta(\xi_\theta(\tau)) = c$ for all $\theta$ and some constant $c$.
    Under these conditions, we can calculate the asymptotic variance of our estimator, which has a minimum when $U$ is in the center of mass of the sampled points.

    \begin{proposition}  \label{prop:distribution}
        Let $\xi_\theta(\tau)$ be the $\tau$ quantile of $g_\theta(x)$. Let $g_i := g_i(\xi_i(\tau))$ be the probability density of $\xi_\theta(\tau)$ for $\theta = x_i$. Under general regulatory conditions~\citep{Koenker2013}, if $g_i = c$ for all $x_i$ and some constant $c$, then the optimal sampling for the quantile regression procedure is achieved when $\bar{x} = U$, and $\hat{U}$ has the asymptotic distribution:

        $$
        \sqrt{n}(\hat{U} - U) \xrightarrow{d} N\left(0, \frac{\tau(1-\tau)}{c^2} \cdot \frac{1}{\beta_1^2}\right) .
        $$
    \end{proposition}

    \begin{proof}
        First, let $D = \lim\limits_{n\to\infty}n^{-1} \sum_{i=1}^{n}x_i{x_i}^T$.
        The asymptotic distribution of $\hat{\beta}$ is~\citep{Koenker2013}:
        $$
        \sqrt{n}(\hat{\beta} - \beta) \xrightarrow{d} N(0, w^2D^{-1}),\,
        w^2 = \frac{\tau(1-\tau)}{c^2} ,
        $$
        If we denote
        $$
        D = \begin{bmatrix}
        1 & \bar{x} \\
        \bar{x} & \overline{x^2}
        \end{bmatrix} ,
        $$
        the inverse is given by

        $$
        D^{-1} = \frac{1}{V(x)}
        \begin{bmatrix}
        \overline{x^2} & -\bar{x} \\
        -\bar{x} & 1
        \end{bmatrix}
        ,  \,
        V(x) = \overline{x^2} - {\bar{x}}^2 .
        $$

        The estimate of the endpoint is given by  $\hat{U} = \frac{\hat{\theta} - \hat{\beta}_0}{\hat{\beta}_1}$.
        Now, define $h(\beta_0, \beta_1) = \frac{\hat{\theta} - \beta_0}{\beta_1}$,
        and according to the Delta method
        $$
        \sqrt{n}(\hat{U} - U) \xrightarrow{d} N(0, w^2\cdot \bigtriangledown h^T \cdot D^{-1} \cdot \bigtriangledown h) ,
        $$
        Where $\bigtriangledown h$ is the gradient of h:
        $\bigtriangledown h = \left(\frac{-1}{\beta_1}, \frac{\beta_0 - \hat{\theta}}{\beta_1^2} \right)$. Multiplying the elements of the variance we get

        $$
        \lim\limits_{n \to \infty}\mathrm{Var}(\sqrt{n} \cdot \hat{\theta}) = \frac{w^2}{V(x)\beta_1^2}(\overline{x^2} - 2\bar{x}U + U^2) ,
        $$
        which achieves a minimum whenever $U = \bar{x}$, So the optimal sampling for the quantile regression would be such that the true endpoint is in the center of mass of the sample points.
        Under the optimal sampling we get a variance of

        $$
        \lim\limits_{n \to \infty}\mathrm{Var}(\sqrt{n} \cdot \hat{\theta}) = \frac{w^2}{V(x)\beta_1^2} \cdot V(x) = \frac{w^2}{\beta_1^2} ,
        $$
        and therefore
        $$
        \sqrt{n}(\hat{U} - U) \xrightarrow{d} N \left(0, \frac{\tau(1-\tau)}{c^2} \cdot \frac{1}{\beta_1^2}\right) .
        $$

    \end{proof}

    \begin{proposition}
        If g belongs to a location family, that is $g_\theta(x) = g(x-\theta)$ for a common density function $g$, then $g_i(\xi_i(\tau))$ are fixed, $\beta_1 = 1$ and the asymptotic variance is equal to $\frac{\tau(1-\tau)}{c^2}$, the variance of the Robbins-Monro estimate.

    \end{proposition}

    \begin{proof}
        Let $G_{x_i}(x)$ be the distribution function of $\hat{\Theta}(X)   $ given $\theta=x_i$, and let $G$ be the standard distribution function of the family.

        To see that the quantile is linear with slope 1 notice that
        $$
        G_{x_i}(\xi_i(\tau)) = G(\xi_i(\tau) - x_i) = \tau .
        $$

        Now, for a continuous G
        $$
        \xi_i(\tau) - x_i = G^{-1}(\tau) \implies \xi_i(\tau) = G^{-1}(\tau) + x_i ,
        $$
        and $\xi_i(\tau)$ is indeed linear in $x_i$ with slope 1 and density:
        $$
        g_i(\xi_i(\tau)) = g_i(G^{-1}(\tau)+x_i) = g(G^{-1}(\tau) + x_i - x_i) = g(G^{-1}(\tau)) ,
        $$
        which does not depend on $x_i$.

        Finally, we can use \Cref{prop:distribution} to get that:

        $$
        \sqrt{n}(\hat{U} - U) \xrightarrow{d} N \left(0,  \frac{\tau(1-\tau)}{c^2} \right) .
        $$

    \end{proof}

    The efficiency of the RM procedure depends on an optimal determination of the step size constant, which is unknown but can be estimated adaptively. The efficiency of the quantile regression method depends on the centering of the data points around the true endpoint, which is again unknown but estimated adaptively.
    For a large number of iterations the two methods indeed perform comparably, but for small sample sizes we show that the adaptive quantile regression gives lower MSEs while using the same number of iterations.

    \subsubsection{Suggested Algorithm}
    The asymptotic results from the last section suggest that the lowest variance should be reached when the true endpoint is in the center of the data. Consequently, we suggest estimating the true endpoint by the current estimate at each iteration and to sample a new $x$ that centers the data around it, as described in \Cref{algo:qauntreg}.

    \begin{algorithm}

        Choose a grid of points $\theta^*=(\theta_1, ..., \theta_s)$

        Initialize a list $\hat{\theta}^*$ = []

        \For{$j\leftarrow 1$ \KwTo $s$}{
        Sample $x_j$ from $f_{\theta_j}$

        $\hat{\theta}_j \leftarrow \hat{\Theta}(x_j)$

        Append $\hat{\theta}_j$  to $\hat{\theta}^*$
    }
        \For{$i\leftarrow 1$ \KwTo $N$}{

            Estimate the conditional quantile function: $\hat{Q}_\alpha(\theta) = \hat{\beta}_0 + \hat{\beta}_1\theta$, based on $\theta^*$ and $\hat{\theta}^*$

            $U_i \leftarrow \frac{\hat{\theta} - \hat{\beta}_0}{\hat{\beta}_1}$.

            $\theta_{s+i} \leftarrow (s+i)\cdot U_i - \sum_{j=1}^{s+i-1}\theta_j$

            Append $\theta_{s+i}$ to $\theta^*$

            Calculate $\hat{\theta}_{s+i}$ based on a bootstrap sample with $\theta=\theta_{s+i}$

            Append $\hat{\theta}_{s+i}$ to $\hat{\theta}^*$
        }

        $\hat{U} \leftarrow U_N$

    \caption{Finding the upper endpoint of a $100(1-\alpha)\%$ confidence interval by adaptive quantile regression}

    \label{algo:qauntreg}
    \end{algorithm}

    Notice that if a two sided confidence interval is needed, the points sampled for the upper end can be used in the search for lower end, decreasing the variance further.

    \subsection{Measuring Accuracy and Linearity}

    Two things should be in mind when considering the validity of our approach: The correctness of the linear model, and the Monte-Carlo error that results from using finitely many sample points.
    For testing linearity, we offer to simply test the model against a polynomial model of higher order using a likelihood ratio test or a Wald test as described in \cite{Koenker2013}. If the linear model is correct, then the extra terms should not be significant. We describe our suggestion for the non linear case in the next section.

    For assessing the magnitude of error resulting from the procedure, we offer a simple approach for constructing confidence intervals for the upper end point $U$ that is based on the inversion of the quantile regression confidence interval.

    \cite{Koenker2013} describes and compares different methods to construct confidence intervals for the quantile regression prediction, $Q_\alpha(x)$. Any of them can be inverted in order to get confidence intervals for the endpoint $U$.
    The confidence intervals given by Koenker give us CIs for $Q_\alpha(x)$. In order to construct confidence intervals for the quantity of interest $U$ they have to be inverted, this is accomplished by the same inversion methodology that is used in order to calculate $\hat{U}$. Essentially, the CIs for $Q_\alpha(x)$ are the acceptance region for $U$ and thus, the point where the lower bound of the quantile regression is equal to $\hat{\theta}$ is the upper bound for $U$, as made precise by \Cref{prop:ci_ci}.

    \begin{proposition}
        Let $\hat{Q}_\alpha(\theta) = \hat{\beta}_0 + \hat{\beta}_1\cdot \theta$ be the quantile regression line for the $\alpha$ quantile. Denote $Q_U$ as the point for which $P(\hat{Q}_\alpha(Q_U) < \hat{\theta}) = \delta$ for some $0 \leq \delta \leq 1$. Then $P(\hat{U} > Q_U) = \delta$, that is, $Q_U$ is the $\delta$ level upper bound for $U$.

        \label{prop:ci_ci}
    \end{proposition}

    \begin{proof}
        By definition, $\hat{Q}_\alpha(\theta) = \hat{\beta}_0 + \hat{\beta}_1 \cdot \theta$. Therefore:
        $$
            \delta = P(\hat{Q}_\alpha(Q_U) < \hat{\theta}) = P(\hat{\beta}_0 + \hat{\beta}_1 \cdot Q_U < \hat{\theta}) = P\left(Q_U < \frac{\hat{\theta} - \hat{\beta}_0}{\beta_1}\right) ,
        $$
        but $\hat{U} = \frac{\hat{\theta} - \hat{\beta}_0}{\beta_1}$, and so:
        $$
         \delta = P\left(Q_U < \frac{\hat{\theta} - \hat{\beta}_0}{\beta_1}\right) = P(\hat{U} > Q_U).
        $$
    \end{proof}

    \section{Evaluation}

    \subsection{When the linear model is true}
    In order to compare the methods on simulated data, we implemented the RM algorithm with an adaptive determination of the step size constant as described in \cite{PaulH.GarthwaiteandStephenT.Buckland2016}. We used the same method for both RM (Robbins-Monro) and BRM (Binary Robbins-Monro).
    In order to determine the starting point for the RM and BRM methods we randomly chose a point from a normal distribution centered around the correct endpoint, as done in \cite{Joseph2004}.

    We start by comparing the three methods (RM, BRM, AQR) for the purpose of constructing a one-sided confidence interval for four common estimates. The mean of a normal distribution, the standard deviation of a normal distribution, the shape of a gamma distribution and the mean of a logistic distribution. The root mean squared error of the different methods is shown in \Cref{fig:good_simultaions}.

    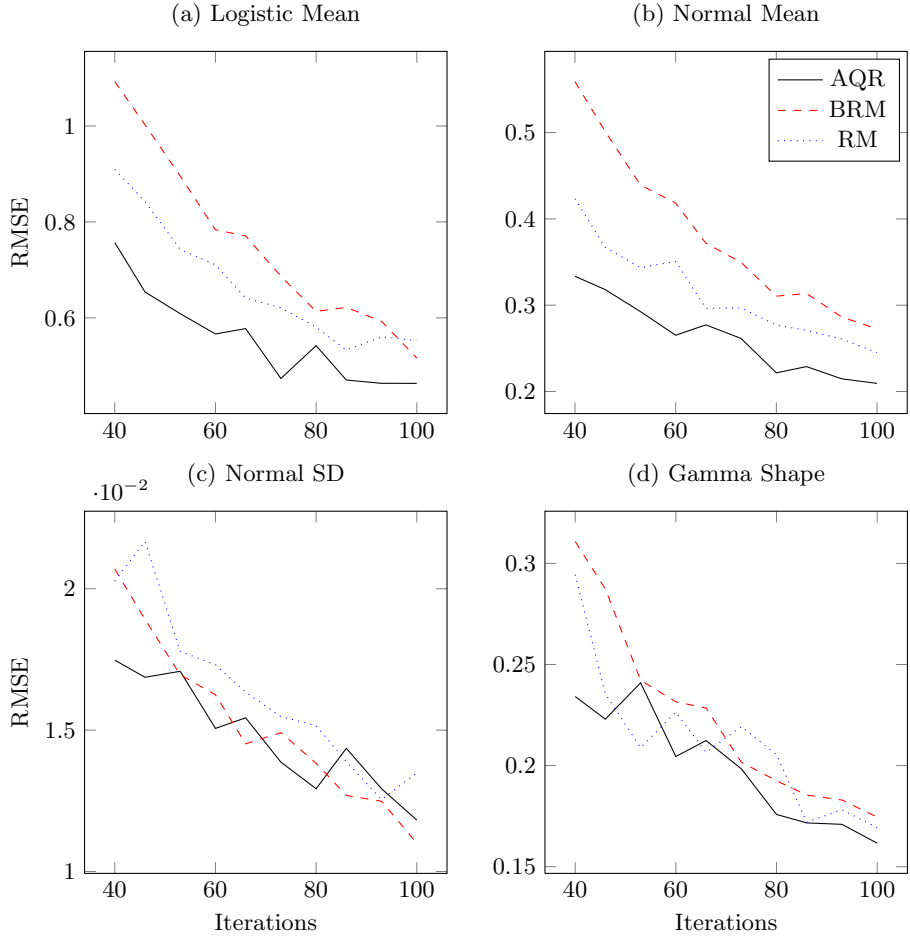
\begin{figure}
        \begin{tikzpicture}
\begin{axis}[
title=(a) Logistic Mean,
height = 6.4cm, width=6.4cm,
name=nw,
ylabel=RMSE]
\addplot[color=black] coordinates {
	( 40 , 0.7570677 ) 
	( 46 , 0.6541243 ) 
	( 53 , 0.6087789 ) 
	( 60 , 0.5664907 ) 
	( 66 , 0.5776405 ) 
	( 73 , 0.473694 ) 
	( 80 , 0.5420933 ) 
	( 86 , 0.4704652 ) 
	( 93 , 0.4635873 ) 
	( 100 , 0.4635075 ) 
};

\addplot[color=red, dashed] coordinates {
	( 40 , 1.092513 ) 
	( 46 , 1.002498 ) 
	( 53 , 0.8958049 ) 
	( 60 , 0.7835433 ) 
	( 66 , 0.7707501 ) 
	( 73 , 0.686801 ) 
	( 80 , 0.6135859 ) 
	( 86 , 0.621773 ) 
	( 93 , 0.592727 ) 
	( 100 , 0.5161826 ) 
};

\addplot[color=blue, dotted] coordinates {
	( 40 , 0.9093762 ) 
	( 46 , 0.842583 ) 
	( 53 , 0.7428954 ) 
	( 60 , 0.7105772 ) 
	( 66 , 0.6413429 ) 
	( 73 , 0.6212259 ) 
	( 80 , 0.5804614 ) 
	( 86 , 0.5330147 ) 
	( 93 , 0.5604354 ) 
	( 100 , 0.5526607 ) 
};

\end{axis}

\begin{axis}[
title=(b) Normal Mean,
at={($(nw.east)+(1.3cm,0)$)}, anchor=west, height = 6.4cm, width=6.4cm]

\addplot[color=black] coordinates {
( 40 , 0.3335176 ) 
( 46 , 0.3181102 ) 
( 53 , 0.2926253 ) 
( 60 , 0.2652139 ) 
( 66 , 0.2772224 ) 
( 73 , 0.2614879 ) 
( 80 , 0.2216325 ) 
( 86 , 0.2288666 ) 
( 93 , 0.2146825 ) 
( 100 , 0.2094354 ) 
};

\addlegendentry{AQR}

\addplot[color=red, dashed] coordinates {
	( 40 , 0.5591433 ) 
	( 46 , 0.501417 ) 
	( 53 , 0.4393057 ) 
	( 60 , 0.4183443 ) 
	( 66 , 0.3717805 ) 
	( 73 , 0.349642 ) 
	( 80 , 0.3105219 ) 
	( 86 , 0.3135361 ) 
	( 93 , 0.2863068 ) 
	( 100 , 0.2725241 ) 
};
\addlegendentry{BRM}

\addplot[color=blue, dotted] coordinates {
	( 40 , 0.4234334 ) 
	( 46 , 0.3669267 ) 
	( 53 , 0.3435765 ) 
	( 60 , 0.3507568 ) 
	( 66 , 0.2963637 ) 
	( 73 , 0.2968345 ) 
	( 80 , 0.276984 ) 
	( 86 , 0.2708163 ) 
	( 93 , 0.2608199 ) 
	( 100 , 0.2451416 ) 
};
\addlegendentry{RM}
\end{axis}

\begin{axis}[
name=sw,
title=(c) Normal SD,
at={($(nw.south)+(0,-1.3cm)$)}, anchor=north, height = 6.4cm, width=6.4cm,
xlabel=Iterations,
ylabel=RMSE]

\addplot[color=black] coordinates {
	( 40 , 0.01747514 ) 
	( 46 , 0.01686856 ) 
	( 53 , 0.0170804 ) 
	( 60 , 0.01506089 ) 
	( 66 , 0.01543481 ) 
	( 73 , 0.01386702 ) 
	( 80 , 0.01292699 ) 
	( 86 , 0.01435504 ) 
	( 93 , 0.01292635 ) 
	( 100 , 0.01181893 ) 
};

\addplot[color=red, dashed] coordinates {
	( 40 , 0.02069632 ) 
	( 46 , 0.01891923 ) 
	( 53 , 0.01695601 ) 
	( 60 , 0.0162492 ) 
	( 66 , 0.0145147 ) 
	( 73 , 0.01491003 ) 
	( 80 , 0.01382853 ) 
	( 86 , 0.01268845 ) 
	( 93 , 0.012485 ) 
	( 100 , 0.0110001 ) 
};

\addplot[color=blue, dotted] coordinates {
	( 40 , 0.02026935 ) 
	( 46 , 0.02167374 ) 
	( 53 , 0.01778594 ) 
	( 60 , 0.01730979 ) 
	( 66 , 0.01634211 ) 
	( 73 , 0.01548114 ) 
	( 80 , 0.01515896 ) 
	( 86 , 0.01388369 ) 
	( 93 , 0.01254316 ) 
	( 100 , 0.01348942 ) 
};

\end{axis}

\begin{axis}[
title=(d) Gamma Shape ,
at={($(sw.east)+(1.3cm,0)$)}, anchor=west, height = 6.4cm, width=6.4cm,
xlabel=Iterations]

\addplot[color=black] coordinates {
	( 40 , 0.2341691 ) 
	( 46 , 0.2229783 ) 
	( 53 , 0.2409413 ) 
	( 60 , 0.2044859 ) 
	( 66 , 0.2123854 ) 
	( 73 , 0.1984245 ) 
	( 80 , 0.1758362 ) 
	( 86 , 0.1716512 ) 
	( 93 , 0.1709577 ) 
	( 100 , 0.1616169 ) 
};

\addplot[color=red, dashed] coordinates {
	( 40 , 0.3108724 ) 
	( 46 , 0.2873334 ) 
	( 53 , 0.2421934 ) 
	( 60 , 0.2315572 ) 
	( 66 , 0.2284827 ) 
	( 73 , 0.2017036 ) 
	( 80 , 0.1924742 ) 
	( 86 , 0.1854545 ) 
	( 93 , 0.1829803 ) 
	( 100 , 0.174496 ) 
};

\addplot[color=blue, dotted] coordinates {
	( 40 , 0.2944362 ) 
	( 46 , 0.2354079 ) 
	( 53 , 0.2088559 ) 
	( 60 , 0.2264575 ) 
	( 66 , 0.206537 ) 
	( 73 , 0.2193266 ) 
	( 80 , 0.2053649 ) 
	( 86 , 0.1720365 ) 
	( 93 , 0.1781482 ) 
	( 100 , 0.1692832 ) 
};

\end{axis}

\end{tikzpicture}
        \caption{The root mean squared error of the Robbins-Monro (RM), Binary Robbins-Monro (BRM) and Adaptive Quantile Regression (AQR) for estimating the upper endpoint of a CI for \textbf{(a)} the mean of a standard logistic distribution \textbf{(b)} the mean of a N(0, 0.1) distribution \textbf{(c)} the standard deviation of a N(0, 0.1) distribution \textbf{(d)} the shape parameter of a Gamma(10, 1) distribution. AQR is more accurate given 40 iterations in all simulations and remains more accurate for up to 100 simulation for the logistic and normal mean estimators.}
        \label{fig:good_simultaions}
    \end{figure}

    \Cref{fig:good_simultaions} shows that with 40 iterations, the AQR procedure is more accurate for all four distributions, and it remains more accurate with up to 100 iterations for the estimation of the mean of the logistic and normal distribution.

    However, regression procedures are usually sensitive to outliers and AQR is no exception, as verified by measuring the error in estimating the endpoint of a CI for the shape of a Gamma(5,1) distribution, shown in \Cref{fig:gamma_5}.

    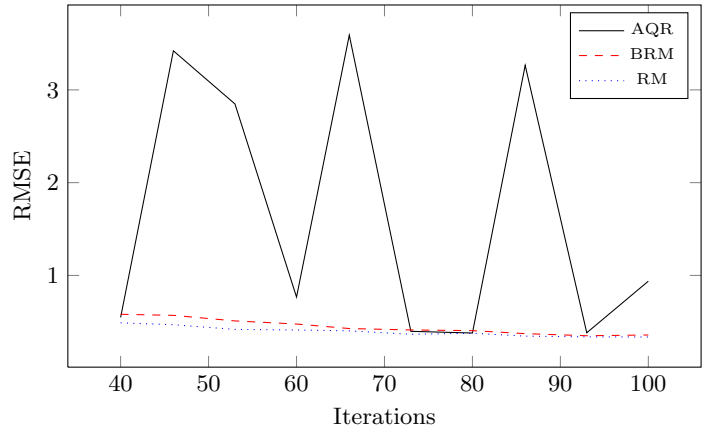
\begin{figure}
        \begin{tikzpicture}

\begin{axis}
[height = 6.4cm, width=10cm,
xlabel=Iterations,
ylabel=RMSE,
legend style = {font = \tiny}]

\addplot[color=black] coordinates {
( 40 , 0.547262 ) 
( 46 , 3.421273 ) 
( 53 , 2.848314 ) 
( 60 , 0.7666576 ) 
( 66 , 3.589957 ) 
( 73 , 0.3972672 ) 
( 80 , 0.3790674 ) 
( 86 , 3.264053 ) 
( 93 , 0.381777 ) 
( 100 , 0.9377178 ) 
};
\addlegendentry{AQR}

\addplot[color=red, dashed] coordinates {
( 40 , 0.5809274 ) 
( 46 , 0.5692915 ) 
( 53 , 0.5095134 ) 
( 60 , 0.4760328 ) 
( 66 , 0.4272468 ) 
( 73 , 0.4122189 ) 
( 80 , 0.4046629 ) 
( 86 , 0.3724643 ) 
( 93 , 0.3495907 ) 
( 100 , 0.3578571 )  	
};
\addlegendentry{BRM}

\addplot[color=blue, dotted] coordinates {
( 40 , 0.4889279 ) 
( 46 , 0.4685582 ) 
( 53 , 0.4179734 ) 
( 60 , 0.4118114 ) 
( 66 , 0.4021339 ) 
( 73 , 0.3653558 ) 
( 80 , 0.3767114 ) 
( 86 , 0.3454104 ) 
( 93 , 0.3375941 ) 
( 100 , 0.3364234 ) 
};
\addlegendentry{RM}
\end{axis}
\end{tikzpicture}
        \caption{Comparison by root MSE of adaptive quantile regression (AQR), Robbins-Monro (RM) and Binary Robbins-Monro (BRM) methods for determination of the upper endpoint of a confidence interval for a Gamma(5,1) distribution.}
        \label{fig:gamma_5}
    \end{figure}

    The results of this section suggest that for distributions that are not prone to outliers, and when only a small number of iterations is possible, AQR gives a better accuracy than the other methods.

    \subsection{When the linear model is false}

    The quantile regression methodology that we described relies on the linear model being true. In cases when the conditional quantile function of interest is not linear it is not expected to work properly. a mild non-linearity, such as the one of the binomial distribution with $n=30$, will hurt the accuracy of the method somewhat, as demonstrated in \Cref{fig:binomial}.

    \begin{figure}
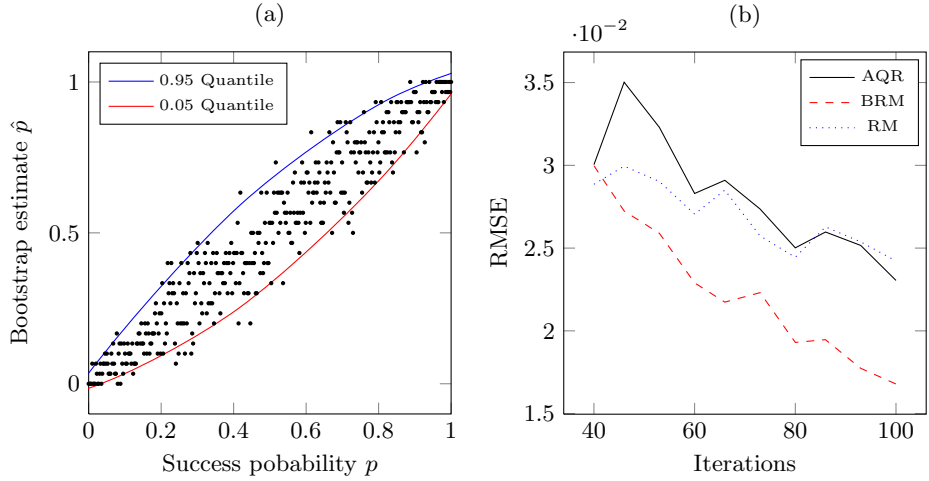

        \include{figure_binomial_dist}
        \caption{\textbf{(a)} The non linear quantiles of an estimate for the success probability in the binomial distribution with $n=30$. \textbf{(b)} Comparison of adaptive quantile regression (AQR), Robbins-Monro (RM) and Binary Robbins-Monro (BRM) methods for determination of the upper endpoint of a confidence interval for p by their root MSE.}
        \label{fig:binomial}
    \end{figure}

    A more severe non linearity will cause the linear modeling to give very inaccurate results. We show this with a simulation using an $N(\theta, ((1+\theta^2)/8)^2)$ distribution, where the goal is to construct a CI for $\theta$. The results are shown in \Cref{fig:non_linear}.

    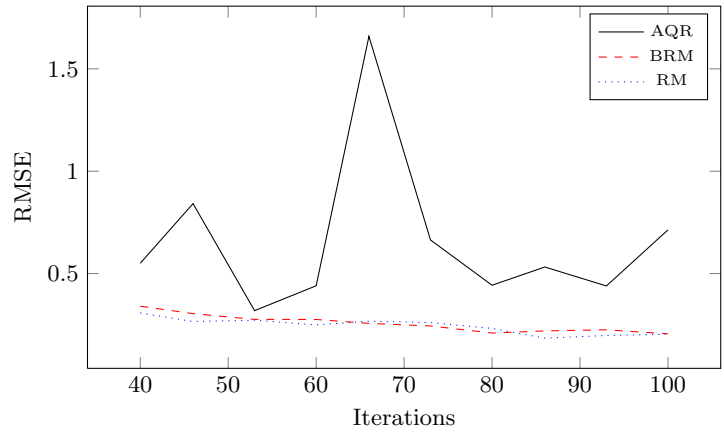
\begin{figure}
        \centering
        \begin{tikzpicture}

\begin{axis}
[height = 6.4cm, width=10cm,
xlabel=Iterations,
ylabel=RMSE,
legend style = {font = \tiny}]

\addplot[color=black] coordinates {
( 40 , 0.5503392 ) 
( 46 , 0.8416941 ) 
( 53 , 0.3188305 ) 
( 60 , 0.4403835 ) 
( 66 , 1.658739 ) 
( 73 , 0.6640667 ) 
( 80 , 0.4430132 ) 
( 86 , 0.5320924 ) 
( 93 , 0.4394383 ) 
( 100 , 0.7133102 ) 
};
\addlegendentry{AQR}

\addplot[color=red, dashed] coordinates {
( 40 , 0.3405423 ) 
( 46 , 0.3042753 ) 
( 53 , 0.2762472 ) 
( 60 , 0.2756966 ) 
( 66 , 0.2566476 ) 
( 73 , 0.2440453 ) 
( 80 , 0.2095168 ) 
( 86 , 0.2199306 ) 
( 93 , 0.2252429 ) 
( 100 , 0.2054483 )
};
\addlegendentry{BRM}

\addplot[color=blue, dotted] coordinates {
( 40 , 0.3075648 ) 
( 46 , 0.265793 ) 
( 53 , 0.2722264 ) 
( 60 , 0.249198 ) 
( 66 , 0.2672979 ) 
( 73 , 0.2604936 ) 
( 80 , 0.2320176 ) 
( 86 , 0.1844735 ) 
( 93 , 0.1974254 ) 
( 100 , 0.204079 ) 
};
\addlegendentry{RM}
\end{axis}
\end{tikzpicture}
        \caption{Comparison of adaptive quantile regression (AQR), Robbins-Monro (RM) and Binary Robbins-Monro (BRM) methods for determination of the upper endpoint for $\theta$ in $N(\theta, ((1+\theta^2)/8)^2)$ distribution, which has a non linear quantile.}
        \label{fig:non_linear}
    \end{figure}

    Still, the situation can be rectified by testing for linearity as described above, and using one of the following measures when the linearity test fails:

    \begin{enumerate}
        \item Fitting a polynomial model and using its prediction.
        \item Fitting an a-parametric model and using its prediction.
        \item Using the prediction of a polynomial model as a hot-start for a Robbins-Monro search.
    \end{enumerate}

    We recommend the last method, and evaluate it compared to using all the iterations for the stochastic approximation.

    The non-linear quantile is created by simulating a Normal distribution with mean $\theta$ and standard deviation $\theta(5-\theta)/8$. The quantile regression is done by taking 20 evenly spaced points, fitting a quadratic polynomial to them using quantile regression, and using the prediction to start a Robbins-Monro procedure. The Robbins-Monro procedure starts as if it already did 20 steps.
    For the Binary Robbins-Monro procedure, we just started the search from the predicted point and halved the value for the prior variance. We refer to the procedures of using the RM or BRM with the quantile regression prediction as
    a starting point as Hot Started RM and Hot Started BRM, respectively.

    As shown in \Cref{fig:hot_started}, we see practically no loss of accuracy from using the first 20 points to predict a starting point for the RM procedure, after concluding that the linearity assumption does not hold.

    \begin{figure}[H]
        \centering
        \begin{tikzpicture}

\begin{axis}
[height = 6.4cm, width=10cm,
xlabel=Iterations,
ylabel=RMSE,
legend style = {font = \tiny}]

\addplot[color=red, dashed] coordinates {
( 40 , 0.3217083 ) 
( 46 , 0.3189682 ) 
( 53 , 0.2861746 ) 
( 60 , 0.2777733 ) 
( 66 , 0.2546992 ) 
( 73 , 0.2374106 ) 
( 80 , 0.237141 ) 
( 86 , 0.232191 ) 
( 93 , 0.2176928 ) 
( 100 , 0.217047 ) 	
};
\addlegendentry{BRM}

\addplot[color=blue, dotted] coordinates {
( 40 , 0.28774 ) 
( 46 , 0.2861608 ) 
( 53 , 0.276704 ) 
( 60 , 0.2474726 ) 
( 66 , 0.2412809 ) 
( 73 , 0.2345295 ) 
( 80 , 0.2309472 ) 
( 86 , 0.215003 ) 
( 93 , 0.1993267 ) 
( 100 , 0.195022 ) 
};
\addlegendentry{RM}

\addplot[color=orange, densely dotted] coordinates {
	( 40 , 0.3633196 ) 
	( 46 , 0.3792746 ) 
	( 53 , 0.3332031 ) 
	( 60 , 0.2921122 ) 
	( 66 , 0.2825606 ) 
	( 73 , 0.2742651 ) 
	( 80 , 0.248625 ) 
	( 86 , 0.2576523 ) 
	( 93 , 0.2402489 ) 
	( 100 , 0.2181691 ) 
};
\addlegendentry{Hot Started BRM}

\addplot[color=black] coordinates {
	( 40 , 0.2906063 ) 
	( 46 , 0.3048459 ) 
	( 53 , 0.2525326 ) 
	( 60 , 0.2519557 ) 
	( 66 , 0.2529586 ) 
	( 73 , 0.2264071 ) 
	( 80 , 0.2204733 ) 
	( 86 , 0.211576 ) 
	( 93 , 0.2141797 ) 
	( 100 , 0.1965478 ) 
};
\addlegendentry{Hot Started RM}

\end{axis}
\end{tikzpicture}
        \caption{Comparison by RMSE of the Robbins-Monro (RM), Binary Robbins-Monro (BRM), Hot Started Robbins-Monro and Hot Started Binary Robbins-Monro for determination of the upper endpoint of a CI for $\theta$ in a $N(\theta, (\theta(5-\theta)/8)^2)$ distribution, which has a non-linear quantile.}
        \label{fig:hot_started}
    \end{figure}
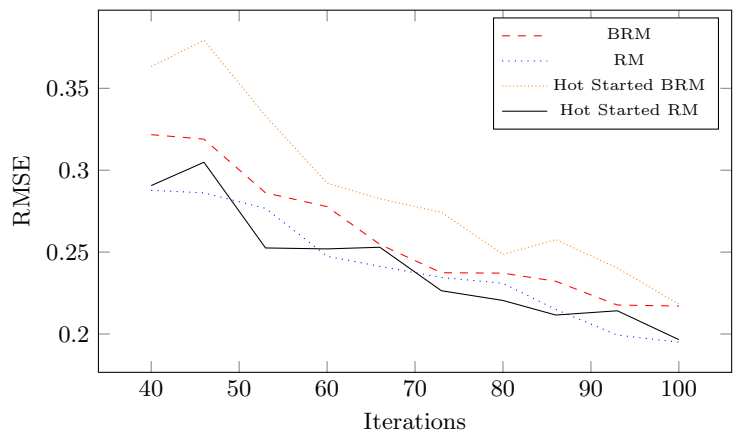

\section{Confidence intervals for demographic parameters in population models}

In this section, we apply the adaptive quantile regression methodology we have described above to the case of demographic parameters in models of population genetics. An important goal in population genetics is the estimation of parameters of population models from present-day DNA sequences, with the intention of learning information about fundamental characteristics of the population, such as its average mutation rate, or about historical and demographic features, such as the historical population size. In addition, these parameters are often used in the estimation of additional quantities of historical or biological interest.

As an example, we focus here on the mutation rate parameter. The mutation rate is defined as the average number of mutations per DNA base-pair, per generation. The mutation rate is of high interest to the genetics community and is used for  the interpretation of mutations implicated in diseases~\citep{Crow2000}, studies of natural selection \citep{McVicker2009} and the study of several aspects of human mutagenesis \citep{Francioli2015}. Moreover, the mutation rate is utilized in the dating of historical events \citep{Li2011}. Different values of the mutation rate significantly impact the dating of such events \citep{Scally2012}. Therefore, it is important to accurately measure the uncertainty in its estimate.

Nevertheless, our current estimates of the human mutation rate vary substantially. A number of different approaches are commonly used in order to estimate the rate, each yielding different results ranging from $1.2\cdot10^{-8}$ for the pedigree based estimators~\citep{Kong2012} to $2.5\cdot10^{-8}$ for the phylogenetic estimators~\citep{Nachman2000}. In this section we focus on the problem of constructing confidence intervals for mutation rate estimates of the second type, and show that the uncertainty in their estimation does not explain their disagreement with the pedigree based estimates.

Due to the complexity of population models, the distribution of mutation rate estimators is usually difficult to characterize or derive analytically. In addition, simulation or sampling from these models is computationally prohibitive. It is therefore important to devise a method that can accurately construct CIs for the mutation rate estimator, while maintaining a feasible computational cost. In the following, we first briefly introduce the Wright-Fisher model and the related coalescent model (for a review, see \cite{Durrett2008}). Finally, we apply the adaptive quantile regression method to construct CIs for the mutation rate under the coalescent model, and show its superiority compared to alternative approaches.

\subsection{The Wright-Fisher model and the Watterson estimator}
The Wright-Fisher model is common model in population genetics. Assume that we have a population of $2N$ individuals, each carrying one of two alleles. According to the Wright-Fisher model, the population in generation $t+1$ is a random choice, with replacement, of $2N$ individuals from generation $t$. To introduce mutations into the model, at each generation, the allele of each offspring is changed with probability $\mu$, where $\mu$ is defined as the mutation rate per individual per generation. Finally, we randomly sample $n$ individuals from the last generation as our observed sample.

We define the scaled mutation rate as $\theta = 4N\mu$. Let $S_n$ be the number of loci in which there is genetic variability in a sample of $n$ individuals. Under the infinite site model~\citep{Kimura1969}, which assumes that mutations do not occur twice at the same locus, it was shown by~\cite{Watterson1975} that the expected value of $S_n$ is
$$
E(S_n) = \theta\cdot\sum_{i=1}^{n-1}\frac{1}{i}\,,
$$
and therefore, that $\theta$ can be estimated by the moment estimator
$$
\hat{\theta} = \frac{s_n}{\sum_{i   =1}^{n-1}\frac{1}{i}}\,.
$$
where $s_n$ is the observed number of variable loci in the sample. \cite{Watterson1975} derived a formula for the variance of $\hat{\theta}$ and showed that it follows an asymptotic normal distribution. However, \cite{Klein1999} showed that the convergence to this distribution is slow and therefore the asymptotic normal distribution is not a good approximation for the Watterson estimator.

\subsection{The coalescent model}

The coalescent model, introduced by \cite{Kingman}, is an approximation to the Wright-Fisher model that makes it possible to conveniently derive analytic results and to sample from the model with lower computational cost. In contrast to the Wright-Fisher model, the coalescent process proceeds backward in time. Beginning with the last generation, a parent is chosen at random from the previous generations for each individual in the sample. If two individuals in the sample selected the same parent, they are said to have coalesced.

Coalescent theory is essential to our application for two main reasons. First, we are only interested in simulating the generations since the most recent common ancestor of the sample, because previous generations do not influence the genetic diversity of the sample. As the coalescent model goes backwards in time, it focuses only on the genealogy of the observed sample, rather than the entire unobserved population. Second, it allows for a much faster simulation, as we can sample the time to the next coalescence event, and ignore the generations in between. For more details on the coalescent simulation, see \cite{msprime}.

\subsection{Constructing confidence intervals for the mutation rate parameter}

In this section, we apply the adaptive quantile regression methodology in order to construct confidence intervals for the mutation rate parameter under the coalescent model and show that it gives more accurate results than the methods based on stochastic approximation.

In order to determine the correct endpoint of a confidence interval to an estimate of $\hat{\mu} = 2.5\cdot10^{-8}$, as estimated by~\cite{Nachman2000}, we simulated 10,000 samples with $n=\text{1,000}$, under the coalescent model with $N=\text{10,000}$ and a genome size of 1,000,000. For the mutation rate, we choose a grid of 50 points in $[-8.5, -6]$ on a log scale. We used the resulting estimates to calculate the sample quantiles for each value of the mutation rate on the grid, and interpolated the sample quantiles (see \Cref{fig:coalescence} (a)). Finally, the endpoint of the confidence interval is calculated using the test-inversion methodology. The resulting 95\% confidence interval is $(1.92\cdot10^{-8},\, 3.32\cdot10^{-8})$, which does not include the pedigree based estimates.

We then applied the adaptive quantile regression (AQR), the Robbins-Monro (RM) and the Binary Robbins-Monro (BRM) methods to estimate the correct endpoint for the confidence interval. We ran each method for 40,50,...,100 iterations, using 300 simulations for each setup. The accuracy is measured by the root mean square error from the correct endpoint presented above. We used msprime \citep{msprime} to simulate the coalescent, and the Watterson estimator in order to estimate the mutation rate.

\begin{figure}
    \begin{tikzpicture}
\begin{axis}[
title=(a),
height = 6.4cm, width=6.4cm,
name=lines,
xlabel=Mutation rate $\theta$,
ylabel=Bootstrap estimate $\hat{\theta}$,
legend style = {legend pos = north west, font = \tiny}]
\addplot[color=black] coordinates {
( -8.5 , -8.5 ) 
( -8.45 , -8.45 ) 
( -8.4 , -8.4 ) 
( -8.35 , -8.35 ) 
( -8.3 , -8.3 ) 
( -8.25 , -8.25 ) 
( -8.2 , -8.2 ) 
( -8.15 , -8.15 ) 
( -8.1 , -8.1 ) 
( -8.05 , -8.05 ) 
( -8 , -8 ) 
( -7.95 , -7.95 ) 
( -7.9 , -7.9 ) 
( -7.85 , -7.85 ) 
( -7.8 , -7.8 ) 
( -7.75 , -7.75 ) 
( -7.7 , -7.7 ) 
( -7.65 , -7.65 ) 
( -7.6 , -7.6 ) 
( -7.55 , -7.55 ) 
( -7.5 , -7.5 ) 
( -7.45 , -7.45 ) 
( -7.4 , -7.4 ) 
( -7.35 , -7.35 ) 
( -7.3 , -7.3 ) 
( -7.25 , -7.25 ) 
( -7.2 , -7.2 ) 
( -7.15 , -7.15 ) 
( -7.1 , -7.1 ) 
( -7.05 , -7.05 ) 
( -7 , -7 ) 
( -6.95 , -6.95 ) 
( -6.9 , -6.9 ) 
( -6.85 , -6.85 ) 
( -6.8 , -6.8 ) 
( -6.75 , -6.75 ) 
( -6.7 , -6.7 ) 
( -6.65 , -6.65 ) 
( -6.6 , -6.6 ) 
( -6.55 , -6.55 ) 
( -6.5 , -6.5 ) 
( -6.45 , -6.45 ) 
( -6.4 , -6.4 ) 
( -6.35 , -6.35 ) 
( -6.3 , -6.3 ) 
( -6.25 , -6.25 ) 
( -6.2 , -6.2 ) 
( -6.15 , -6.15 ) 
( -6.1 , -6.1 ) 
( -6.05 , -6.05 ) 
};
\addlegendentry{Mean Value}

\addplot[color=red, dashed] coordinates {
	( -8.5 , -8.375851 ) 
	( -8.45 , -8.328529 ) 
	( -8.4 , -8.275357 ) 
	( -8.35 , -8.228248 ) 
	( -8.3 , -8.175397 ) 
	( -8.25 , -8.128891 ) 
	( -8.2 , -8.080545 ) 
	( -8.15 , -8.029824 ) 
	( -8.1 , -7.981649 ) 
	( -8.05 , -7.931162 ) 
	( -8 , -7.881271 ) 
	( -7.95 , -7.831778 ) 
	( -7.9 , -7.778992 ) 
	( -7.85 , -7.728576 ) 
	( -7.8 , -7.680758 ) 
	( -7.75 , -7.629016 ) 
	( -7.7 , -7.578034 ) 
	( -7.65 , -7.527803 ) 
	( -7.6 , -7.477694 ) 
	( -7.55 , -7.431917 ) 
	( -7.5 , -7.378983 ) 
	( -7.45 , -7.329007 ) 
	( -7.4 , -7.278464 ) 
	( -7.35 , -7.229868 ) 
	( -7.3 , -7.180147 ) 
	( -7.25 , -7.128125 ) 
	( -7.2 , -7.078731 ) 
	( -7.15 , -7.028574 ) 
	( -7.1 , -6.977896 ) 
	( -7.05 , -6.93225 ) 
	( -7 , -6.878317 ) 
	( -6.95 , -6.827927 ) 
	( -6.9 , -6.780774 ) 
	( -6.85 , -6.731386 ) 
	( -6.8 , -6.680166 ) 
	( -6.75 , -6.631526 ) 
	( -6.7 , -6.5799 ) 
	( -6.65 , -6.531002 ) 
	( -6.6 , -6.47944 ) 
	( -6.55 , -6.431428 ) 
	( -6.5 , -6.377903 ) 
	( -6.45 , -6.328611 ) 
	( -6.4 , -6.278069 ) 
	( -6.35 , -6.232641 ) 
	( -6.3 , -6.178791 ) 
	( -6.25 , -6.129502 ) 
	( -6.2 , -6.078682 ) 
	( -6.15 , -6.027106 ) 
	( -6.1 , -5.978172 ) 
	( -6.05 , -5.931402 ) 
};
\addlegendentry{0.95 Quantile}

\addplot[color=blue, dashed] coordinates {
( -8.5 , -8.614687 ) 
( -8.45 , -8.562937 ) 
( -8.4 , -8.511961 ) 
( -8.35 , -8.462962 ) 
( -8.3 , -8.412513 ) 
( -8.25 , -8.363951 ) 
( -8.2 , -8.312478 ) 
( -8.15 , -8.261113 ) 
( -8.1 , -8.211639 ) 
( -8.05 , -8.160051 ) 
( -8 , -8.11092 ) 
( -7.95 , -8.060747 ) 
( -7.9 , -8.013524 ) 
( -7.85 , -7.960215 ) 
( -7.8 , -7.912272 ) 
( -7.75 , -7.861957 ) 
( -7.7 , -7.809147 ) 
( -7.65 , -7.76156 ) 
( -7.6 , -7.710404 ) 
( -7.55 , -7.659715 ) 
( -7.5 , -7.612072 ) 
( -7.45 , -7.55972 ) 
( -7.4 , -7.509516 ) 
( -7.35 , -7.459534 ) 
( -7.3 , -7.409828 ) 
( -7.25 , -7.358885 ) 
( -7.2 , -7.310031 ) 
( -7.15 , -7.261491 ) 
( -7.1 , -7.212326 ) 
( -7.05 , -7.162355 ) 
( -7 , -7.109484 ) 
( -6.95 , -7.059119 ) 
( -6.9 , -7.009241 ) 
( -6.85 , -6.961422 ) 
( -6.8 , -6.912313 ) 
( -6.75 , -6.861813 ) 
( -6.7 , -6.807531 ) 
( -6.65 , -6.762146 ) 
( -6.6 , -6.710558 ) 
( -6.55 , -6.658737 ) 
( -6.5 , -6.60928 ) 
( -6.45 , -6.561784 ) 
( -6.4 , -6.510702 ) 
( -6.35 , -6.459744 ) 
( -6.3 , -6.410723 ) 
( -6.25 , -6.361589 ) 
( -6.2 , -6.308691 ) 
( -6.15 , -6.257903 ) 
( -6.1 , -6.210915 ) 
( -6.05 , -6.160609 ) 	
};
\addlegendentry{0.05 Quantile}
\end{axis}

\begin{axis}[
title=(b),
at={($(lines.east)+(1.5cm,0)$)}, anchor=west, height = 6.4cm, width=6.4cm,
xlabel=Iterations,
ylabel=RMSE,
legend style = {font = \tiny}]

\addplot[color=black] coordinates {
	( 40 , 0.01680123 ) 
	( 52 , 0.01602657 ) 
	( 64 , 0.013932 ) 
	( 76 , 0.01339286 ) 
	( 88 , 0.01208722 ) 
	( 100 , 0.01145992 ) 
};
\addlegendentry{AQR}

\addplot[color=red, dashed] coordinates {
	( 40 , 0.02270249 ) 
	( 52 , 0.02233038 ) 
	( 64 , 0.01921365 ) 
	( 76 , 0.01834187 ) 
	( 88 , 0.01699659 ) 
	( 100 , 0.01685845 ) 	
};
\addlegendentry{BRM}

\addplot[color=blue, dotted] coordinates {
	( 40 , 0.02785849 ) 
	( 52 , 0.01991233 ) 
	( 64 , 0.01772194 ) 
	( 76 , 0.01717974 ) 
	( 88 , 0.01609097 ) 
	( 100 , 0.01590946 ) 
};
\addlegendentry{RM}
\end{axis}
\end{tikzpicture}
    \caption{\textbf{(a)} The approximately linear quantiles of the Watterson estimator. \textbf{(b)} Comparison of adaptive quantile regression (AQR), Robbins-Monro (RM) and Binary Robbins-Monro (BRM) methods for determination of $U$  by their root MSE , as calculated from 300 simulations.}
    \label{fig:coalescence}
\end{figure}
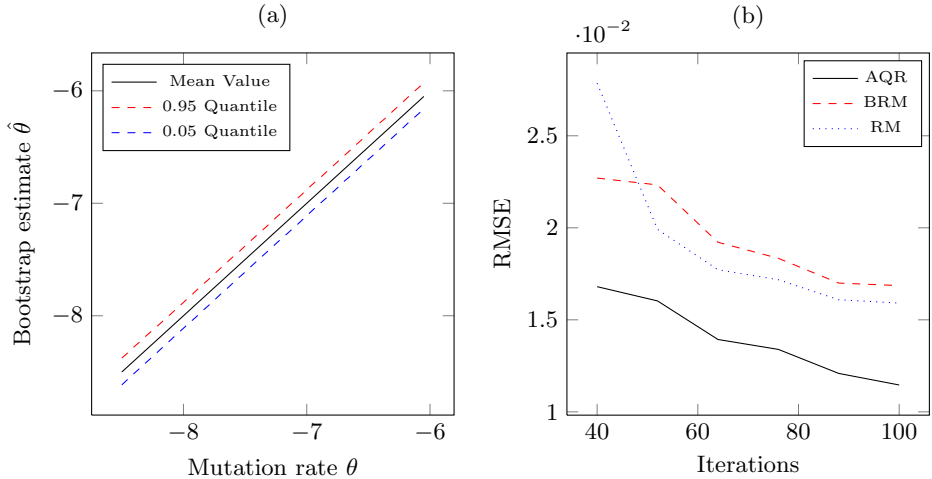

The results, shown in Figure \ref{fig:coalescence}, show a uniform superiority of the adaptive quantile regression methodology over the stochastic approximation based approaches. For 40 iterations, the root mean square error of the AQR method is 40\% lower than that of RM, and 20\% lower than that of BRM, and it remains lower by at least 20\% for 40-100 iterations.

The significantly lower error of the adaptive quantile regression procedure suggests that a real benefit can be achieved by using quantile regression modeling instead of a stochastic approximation when estimating the endpoint of a test-inversion bootstrap CI, making it possible to achieve a desired level of accuracy by using less iterations, thus saving computational resources.

\section{Discussion}

In this paper we proposed a novel method for the construction of confidence intervals based on quantile regression, that is often more efficient than the standard approaches based on stochastic approximation, and can be applied to problems of inference in population genetics models.

The proposed method enjoys better accuracy, and a lack of tunable parameters, in contrast to the stochastic optimization methods that depend on setting a constant to an unknown value in order to achieve full asymptotic efficiency.
We also showed that the asymptotic variance under ideal conditions is surprisingly equal for the stochastic approximation methods and the quantile regression method. Indeed, both perform comparably for a large number of iterations.

For a small number of iterations, it appears that the use that our quantile regression method makes of all previous samples allows it to achieve better efficiency in most cases.
It is also likely that using the points sampled in the search for the upper endpoint in the search for the lower endpoint(or vice-versa), possibly with weights, will improve the performance even further, an enhancement that is not possible for the stochastic approximation based methods.

The sensitivity of the quantile regression methodology to outliers may be improved by considering robust alternatives to quantile regression or by heuristic removal of outliers, but these approaches were not tested in this paper.

The method can be extended to distributions with non-linear quantiles by replacing the linear quantile regression with non-linear or non-parametric regression. Multiple methods exist in this domain and their performance remains to be tested.

To conclude, we recommend the usage of the quantile regression approach when the distribution of the estimator does not seem prone to outliers, and the conditional quantile is locally linear, both of these conditions can be checked after sampling a few points.

For distributions that do not meet these requirements we recommend using a stochastic approximation, with adaptive determination of the step size, and a quantile regression prediction as a starting point. The Binary Robbins-Monro seems to be preferable to the classic one mostly if good prior knowledge regarding the endpoint is available.

\bibliography{FastConfidenceIntervals}

\end{document}